\documentclass{ifacconf}

\usepackage{graphicx}      % include this line if your document contains figures
\usepackage{natbib}        % required for bibliography

\usepackage{tabularx}
\usepackage{threeparttable}
\usepackage{mathrsfs}
\usepackage{color}
\usepackage{amssymb}

\usepackage{amsmath}
\usepackage{enumerate}
\usepackage{algorithm}
\usepackage{algorithmic}
\usepackage{bm}
\usepackage{makecell}
\usepackage{multirow}
\usepackage{mathtools}
\usepackage{bbm}

\usepackage{dsfont}
\usepackage{pifont}
\usepackage{amsfonts}
\usepackage{subfigure}

\newtheorem{remark}{\textbf{Remark}}
\newtheorem{theorem}{\textbf{Theorem}}
\newtheorem{lemma}{\textbf{Lemma}}

\newtheorem{assumption}{\textbf{Assumption}}

\begin{document}
\begin{frontmatter}

\title{Topology Inference for Immune System Networks by Using Cell Amount Data\thanksref{footnoteinfo}} 

\thanks[footnoteinfo]{This work was supported by the Knut and Alice Wallenberg (KAW) Foundation, and the Swedish Research
Council (VR).}

\author[Affil_1]{Yushan Li} 
\author[Affil_2]{Rikard Forlin}
\author[Affil_1]{Dimos V. Dimarogonas}
\author[Affil_2]{Petter Brodin}

\address[Affil_1]{Department of Decision and Control Systems, KTH Royal Institute of Technology, Sweden (email: \{yushanl, dimos\}@kth.se).}
\address[Affil_2]{Department of Women’s and Children’s Health, Karolinska Institutet, Sweden (e-mail: \{rikard.forlin, petter.brodin\}@ki.se).}

\begin{abstract}
Recent years have witnessed the advanced development of topology inference research, which helps elucidate the interaction relationships of components in many biological networks. 
This paper focuses on inferring the topology of a group of immune cells, based on the collected data from cell-depletion based experiments. 
The problem is very challenging due to i) the lack of standard analytical models for the cell interactions, 
and ii) the restrictive data availability determined by the huge experiment and time costs. 
% The challenges of this problem come from two aspects: i) there are no standard analytical models that could fully interpret the cell interactions, 
% and ii) currently we can only obtain a pair of data sets in two timepoints due to the huge costs and long time duration. 
To address these issues, we first leverage certain common knowledge and observations on the experiments to characterize three properties on the cell amounts during the interaction process: state non-negativity, ratio-based convergence, and triple signs of topology weights. 
Then, we construct a new model with simple structure and analytical convenience, and obtain sufficient conditions for the model to accommodate all three properties.  
Finally, based on the constructed model, we propose a constrained quadratic programming method to infer the topology from limited number of data pairs. 
Validation on experiment data demonstrate the effectiveness of the proposed method. 
% collabrate
\end{abstract}


\end{frontmatter}
%===============================================================================

\section{Introduction}\label{sec:introduction}

Network systems have been widely used to model many biological networks, such as brain neurons, genes, and proteins networks \citep{barabasi2004network}. 
Topology inference (or identification) has played an important role to understand the intrinsic interaction relations between different entities in the network. 
For example, it can be used to reveal the connectome topology on brain neurons \citep{srivastava2020models}, or interpret the regulatory mechanism of cells against cancer \citep{anastasiadou2018noncoding}. 
In this paper, we focus on inferring the topology of a group of immune cells based on the measured data from cell-depletion experiments. 

% social networks \citep{mao2024social} 

% brain networks \citep{srivastava2020models}

% {\color{red}{
 
% XXXXX 
% (The following paragraphs are aimed to: 
% i) include some related works on inferring the topology of immune network or other types of biological networks; 
% ii) highlight that the existing inference methods cannot be applied to immune networks due to its specificity and the difficulty of conducting the experiment. )

Concerning inferring the topology of network systems from data, numerous works have been developed, e.g., causality-based \citep{9072648}, vector autoregressive based \citep{zaman2020online}, and graph signal processing based \citep{10146241} methods, to name a few. 
In recent years, lots of research have made promising progress in the inference of biological processes especially on gene regulatory networks (GRNs) \citep{badia2023gene}. 
For instance, \cite{tsiantis2018optimality} proposed an inverse optimal control method to identify the underlying optimality principle from time-series data. 
% \cite{daneker2023systems} investigated the parameter identification problem by modeling the biological processes as ordinary differential equations. 
\cite{dong2025data} utilized the Sinkhorn’s algorithm to infer the signs of promotion/inhibition relationships in GRNs. 
\cite{btaf394} designed an optimal transport based method to fit a differential equation model and infer GRNs. 
% \cite{chen2025inferring} Recurrent Autoencoders
% \cite{btaf394} optimal transport 
% \cite{badia2023gene} have summarized a detailed review on how to infer the gene regulatory networks from the genomic, transcriptomic and chromatin information. 
%  mRNA- and protein-expression levels
Despite of the fruitful advances achieved by these works, there is still much left that is difficult to infer.
% For biological networks such as the immune network, the availability of sufficient data is very restrictive. 
Specifically, for the immune system network, 
% The immune networks further 
the difficulties of having both a large biological variation between individuals and technological variation between datasets become evident. 
This makes it difficult to train current models without a large volume of data, a common feature for the aforementioned works. 
Furthermore, models that can take a less granular overview of cell-cell dependencies from sparse data are lacking, as many focus on GRNs within each cell.

Another obstruction that hinders us from inferring the immune network is that the true interaction mechanism of these cells has not yet been fully elucidated. 
Therefore, different from engineering systems that can be described by well-documented dynamical models, there are no universal models for immune networks. 
Luckily, some properties about the interaction process are at least well acknowledged. 
For example, the non-negativity of the cell amount, the triple signs of the topology weights, and the convergence of the cell amounts to a stable baseline after a perturbation, have been established \citep{perelson1997immunology,gunawardena2010models}. 
How to construct an appropriate model that can accommodate these critical properties of the interaction process of immune cells is of great importance. 

% }}

Motivated by the Ockham's razor, 
 % \citep{sober2015ockham}, 
it is meaningful that one can begin with using the simplest linear time-invariant models to fit the data with certain performance guarantees. 
For instance, given a non-negative initial state, the standard consensus model \citep{olfati2007consensus} can guarantee the state evolution of all nodes will converge a common state. 
The scaled consensus model \citep{roy2015scaled} is further proposed to allow for a ratio pattern in the converging state. 
However, these models require the topology weights to be non-negative. 
\cite{altafini2013consensus} investigated the bipartie consensus model specifically considering negative weights, and provided convergence guarantees. 
\cite{aalto2022linear} considered a stochastic linear model for the gene regulatory network and investigated the identifiablity problem from the mean and the covariance of the state distribution. 
% , but lacking strict constraints on actual state level. 
Nevertheless, they still lack non-negative constraints on the actual state level. 
% , but its converging state contains both negative and positive elements under structurally balanced conditions. 

Based on the above observations, this paper aims to construct a simple model to interpret interaction processes and design a method for inferring the immune system's topology. 
% interpret the interaction process and design the topology inference method accordingly.  
The contributions are summarized as follows. 
First, based on the common prior knowledge and experiment observations, 
we formally characterize three properties for the interaction process of immune cells, including the non-negativity of states, the convergence to a relative stable state, and the triple signs of a topology weight.  
Second, we inherit the structure simplicity of traditional consensus models to construct a new nonlinear model with an appropriate physical meaning for the immune network. 
Specifically, sufficient conditions concerning the topology weights and bounds of state ratios are obtained, which guarantee that the required three properties can be met. 
Finally, based on the experiment data, we provide a constrained quadratic programming method to infer the topology. 
Both numerical simulations and experiments verify the effectiveness of the proposed model and method.

The remainder of this paper is organized as follows. 
In Section \ref{sec:model}, the system model is constructed along with feasibility conditions, and the corresponding inference method is also provided. 
Numerical simulations and experiments based on real-life data are conducted in Section \ref{sec:simulation}.
Finally, Section \ref{sec:conclusion} concludes the paper.

\section{Modeling for The Immune Network}\label{sec:model}

Consider that the immune cell network is described by a gragh $\mathcal{G}=\{\mathcal{V},\mathcal{E}\}$, where $\mathcal{V}=\{1,\cdots,n\}$ is the set of different types of cells and $\mathcal{E}$ is the set of connection edges among the cells. 
Specifically, the edge $(i,j)\in\mathcal{E} $ indicates that cell $j$ have influence on cell $i$, and $w_{ij}$ is the connection weight for $(i,j)$. 
Then, $W=[w_{ij}]_{i,j=1}^n \in \mathbb{R}^{n\times n}$ constitutes the topology matrix among the $n$ types of cells.  
% \begin{align}\nonumber
%   \left\{\begin{aligned}
%     &w_{ij}>0:~&&\text{cell}~j~\text{will prompt cell}~i \\
%     &w_{ij}=0:~&&\text{cell}~j~\text{has no influence on cell}~i\\
%     &w_{ij}<0:~&&\text{cell}~j~\text{will inhibit cell}~i
%   \end{aligned}
%   \right..
% \end{align}

\textit{Notations}. In this paper, we denote $\mathbb{R}^n_{>0} $ ($\mathbb{R}^n_{\ge 0} $) as the set of all $n$-dimensional real-value vectors that have positive (non-negative) elements. 
Let $\bm{1}$ and $I_n$ be the all-one vector and $n$-dimensional identity matrix, respectively. 
 % and the set of positive integers. 
The superscript $(\cdot)^\intercal$ denotes the transpose of a matrix, $\operatorname{vec}(\cdot)$ is the vectorized form of a matrix column by column, and $\otimes$ represents the Kronecker product. 
Given a matrix equation $M_1 X M_2=M_3$ where the matrices are with compatible dimensions, the vectorization of the equation satisfies $(M_2^\intercal \otimes M_1)\operatorname{vec}(X)=\operatorname{vec}(M_1 X M_2)=\operatorname{vec}(M_3)$.

\subsection{Principles of Experiments and Data Acquisition}\label{subsec:principle}

% {\color{red}{
% ( XXXXX ~~~The following paragraph about how the experiment works may need be rephrased more sound ~~~~XXXX.)

In the conducted experiments on investigating the dependencies of immune cells, we need to first knock-out a targeted cell type, and then measure the remaining amounts of all cells at some instants. 
However, due to the huge experiment cost and long time process, we can only collect very few samples for each experimental condition. 
% very few samples in each experiment, and the sampling period is not necessarily identical. 
Specifically, we sample the data at $2$ and $20$ hours, respectively, at unstimulated conditions.  
% at multiple stimulations such as lipopolysaccharide (LPS), PMA-Ionomycin, Resiquimod (R848) and unstimulated condition. 
Notice that the sampled data at each timepoint contain massive information about the interactions, e.g., the population of immune proteins, mRNA and other materials. 
This paper focuses on the amounts of the cells and uses them to infer the topology.  

% }}

% When we look at the data, the amounts of different cells are never the same. 
% Hence, directly using the consensus model will be problematic for analysis. 
% A straightforward idea is letting $x_i$ represents the ratio of a kind of cell. 
% However, what we do in practical experiments is to influence the amounts of some types of cells, meaning that both the influenced cell amount and the total cell amount will change. 

% Motivated by the above facts, we consider that the stable status of this system lies in the relative ratio of different cells is stable. 
% % It is based on the assumption that the cell amounts will always reach to a (relatively) stable stage as time goes by. 
% Let $\mu\in\mathbb{R}^n_{>0}$ be a given reference profile, satisfying $\sum_{i=1}^n \mu_i =1$. 
% Then, one critical property that this model should possess is that the steady state needs to satisfy $x_\infty \propto \mu$, i.e., 
% \begin{align}
%   x_j(\infty)/x_i(\infty)=\mu_j/\mu_i,~\forall i,j\in \mathcal{V}. 
% \end{align}

\subsection{Feasibility Analysis of Existing Models}

Let $x_i$ be the amount of cell $i$. 
Based on the common knowledge and experiment evidences on the immune cell network \citep{perelson1997immunology,gunawardena2010models}, 
we observe the following three properties that the network exhibits on the amount level. 
\begin{itemize}
\item \textit{P1): Non-negativity of the state}. During the whole interaction process among the cells, the amounts of all cells should be positive, i.e., 
  \begin{align}
    x_i(k)\ge 0,~\forall i\in\mathcal{V},~k\ge 0. 
  \end{align}
\item \textit{P2): Convergence to relative stable state}. It is acknowledged that for a well-functioned immune network, the cell amounts should remain stable (denoted by $x^\star$) after reacting to counter a virus. 
Specifically, the amounts of different cells are never the same and thus they have a stable percentage. 
Mathematically, this property can be formulated as 
\begin{align}\label{eq:ratio_property}
\lim_{t \to \infty} x(t)=x^\star ~(\|x^\star\|_2<\infty),~~ 
\frac{x_j^\star}{x_i^\star}=\frac{\mu_j}{\mu_i},~\forall i,j,
\end{align}
where $\mu\in\mathbb{R}^n_{>0}$ is the composition (or relative ratio) profile, satisfying $\sum_{i=1}^n \mu_i =1$. 

\item \textit{P3): Triple-attribute of the topology weight}. 
Based on clinic research, the influence of one type of immune cell on the other can be roughly classified into three kinds: prompting, inhibition, or none. 
Mapping these attributes onto the topology weight, it can be formulated as 
% Specifically, the numerical meaning of $w_{ij}$ is 
\begin{align}
  \!\! \left\{\begin{aligned}
    &w_{ij}>0,~&&\text{if cell}~j~\text{prompts cell}~i \\
    &w_{ij}=0,~&&\text{if cell}~j~\text{has no influence on cell}~i\\
    &w_{ij}<0,~&&\text{if cell}~j~\text{inhibits cell}~i
  \end{aligned}
  \right..
\end{align}
\end{itemize}

% Notice that the interaction of the immune cells is an extremely complex process that has not fully understood so far. 
% Hence, it is not our ambition to cover all factors in the process and construct a universal model once for all. 
As discussed in Section \ref{sec:introduction}, the critical limitation in existing linear consensus models lies in the conflict between the state's non-negativity and the triple signs of a topology weight. 
To overcome this dilemma, we construct a new model that slightly breaks the model linearity but preserves the listed three properties, which will be analyzed by using nonlinear Perron-Frobenius theory \citep{Lemmens2012}.

\subsection{The Proposed Model}

Based on the above arguments, we model the system as the following form
\begin{align}\label{eq:x_model}
x(k+1)=\frac{  \mu^\intercal x(k) }{ \mu^\intercal W x(k) }  W x(k)  .
\end{align}
% \textcolor{blue}{
Note that this model is not intended as a first-principles mechanistic description of immune regulation. 
The interaction among immune cells is an extremely complex process that has not been fully understood so far, and it is not our ambition to cover all factors in the immune process. 
Instead, \eqref{eq:x_model} is a coarse-grained model tailored to the experimental cell-composition data and to the inference objective of this work. 
The discrete-time index $k$ represents consecutive observation windows in the experiment. 
% , and thus the dynamics \eqref{eq:x_model} is not an exact discretization of a known continuous-time mechanistic model. 
% }

Compared with the classic linear mapping $x'(k+1)=W x'(k) $ that could represent a fully decentralized interaction, 
the model \eqref{eq:x_model} further introduces a scaling operation on the state (scaled by ${  \mu^\intercal x(k) }/{ \mu^\intercal W x(k) }$) and exhibits certain centralization nature. 
We observe that this point is reasonable because the immune cell system is commonly regarded to be globally regulated in the human body \citep{poon2020whole}. 
%[https://www.sciencedirect.com/science/article/pii/S258900422030701X]
Notice that in this model, 
\begin{align}
  \mu^\intercal x(k+1)=\mu^\intercal  \frac{  \mu^\intercal x(0) }{ \mu^\intercal W x(k) }  W x(k)=\mu^\intercal x(0),
\end{align}
which indicates the weighted sum of $x$ is invariant in the iteration process. 
This invariance property resembles the weighted state sum of a linear consensus process.  
To ease analysis, we introduce $y(k)=x(k)/ \mu^\intercal x(k)$, and then the model \eqref{eq:x_model} is equivalently written as 
\begin{align}\label{eq:y_model}
y(k+1)= F(y(k))=\frac{W y(k)}{\mu^\intercal W y(k) } .
\end{align}
In the subsequent contents, we will mainly focus on model \eqref{eq:y_model} and analyze its convergence. 

\begin{assumption}\label{assu:initial}
The state $y(k)$ is lower bounded by a universal vector $\beta\in\mathbb{R}^n_{>0}$, i.e., $y(k)\ge \beta$ component-wisely. 
\end{assumption}
% Note that by taking $y(0)=x(0)/ \mu^\intercal x(0)$ into $y(0)\ge \beta$, we have $ $
The implication of Assumption \ref{assu:initial} lies in two aspects. 
On the one hand, it indicates that the ratio of a type of cell in the weighted sum of all cells is lowered bounded (i.e., $\frac{x(k)} {\mu^\intercal x(k)} \ge \beta$). 
This point is reasonable because in a healthy immune system, a cell type will have a individual-specific lower bound stemming from both inherited and non-inherited effects, otherwise the immune system will not function well.
On the other hand, it corresponds to the fact that the cell amounts are always nonnegative.  
However, since the topology $W$ contains both negative and non-negative entries, it is possible that not all $W\in\mathbb{R}^{n\times n}$ will satisfy Assumption \ref{assu:initial}.

Next, we will demonstrate under what conditions Assumption \ref{assu:initial} can be met. 
 % which essentially necessitates certain structure properties for the topology. 
Based on the bound vector $\beta$, we have $\mu^\intercal y(0)=1\ge \mu^\intercal \beta$. 
Define an auxiliary residual variable
\begin{align}
  % \mu^\intercal y(0)=1\ge \mu^\intercal \beta \Rightarrow 
  r=1- \mu^\intercal \beta \ge0 .
\end{align}
Then, we define the restricted $\mu$-simplex set as 
\begin{align}
\Delta(\beta)=\{y\in\mathbb{R}^n_{\ge0}: ~\mu^\intercal y=1,~ y\ge\beta\}.
\end{align}
The following result shows how to guarantee that the mapping $F(y)$ is invariant under $\Delta(\beta)$.

\begin{theorem}[Invariance for $F$]\label{th:F_invariance}
Suppose there exist positive constants $b_l$ and $b_u$ such that the following bounds hold for each row of $W$
\begin{equation}\label{eq:lower_upper_bounds}
\left\{\begin{aligned}
c_l(i)&=\sum_{j=1}^n W_{ij}\beta_j
    + r\min_{1\le j\le n}\frac{W_{ij}}{\mu_j}
    \ge b_l~\\
c_u(i)&=\sum_{j=1}^n W_{ij}\beta_j
    + r\max_{1\le j\le n}\frac{W_{ij}}{\mu_j} \le b_u,~
    \frac{b_l }{b_u}\mathbf{1} \ge \beta
\end{aligned}\right..
\end{equation}
Then, for all $y\in\Delta(\beta)$, it holds that 
\begin{align}
b_l \mathbf{1}\le Wy\le b_u \mathbf{1},\qquad F(y)\in\Delta(\beta).
\end{align}
% Specifically, the mapping $F$ admits a unique fixed point $y^\star\in\Delta(\beta)$ satisfying 
% \begin{align}
% Wy^\star=\kappa y^\star,\qquad \langle \mu,y^\star\rangle=1. 
% \end{align}
\end{theorem}

\begin{pf}
First, we prove the state positivity in the dynamic process \eqref{eq:y_model}. 
Since $y\in\Delta(\beta)$, we decompose $y=\beta+\eta$, where $\eta\ge0$ by construction. 
Then, for the $i$-th element of $Wy$, we have
\begin{align}
(Wy)_i &=\sum_{j=1}^n W_{ij}(\beta_j+\eta_j) = \sum_{j=1}^n W_{ij}\beta_j
   + \sum_{j=1}^n \frac{W_{ij}}{\mu_j}(\mu_j\eta_j).
\end{align}
Notice that $r=1- \mu^\intercal \beta= \mu^\intercal (y-\beta)= \mu^\intercal \eta$, and thus $(Wy)_i$ is bounded by
\begin{align}
&\begin{aligned}
(Wy)_i &\ge \sum_{j=1}^n W_{ij}\beta_j +  \min_{1\le j\le n} \frac{W_{ij}}{\mu_j} \sum_{j=1}^n (\mu_j\eta_j) \\
&\ge \sum_{j=1}^n W_{ij}\beta_j +  r \min_{1\le j\le n} \frac{W_{ij}}{\mu_j} \ge b_l,
\end{aligned} \\
&\begin{aligned}
(Wy)_i &\le \sum_{j=1}^n W_{ij}\beta_j +  \max_{1\le j\le n} \frac{W_{ij}}{\mu_j} \sum_{j=1}^n (\mu_j\eta_j) \\
&\le \sum_{j=1}^n W_{ij}\beta_j +  r \max_{1\le j\le n} \frac{W_{ij}}{\mu_j} \le b_u.
\end{aligned}
\end{align}
% \begin{align}
% &\left\{\begin{aligned}
% (Wy)_i &\ge \sum_{j=1}^n W_{ij}\beta_j +  \min_{1\le j\le n} \frac{W_{ij}}{\mu_j} \sum_{j=1}^n (\mu_j\eta_j) \\
% (Wy)_i &\le \sum_{j=1}^n W_{ij}\beta_j +  \max_{1\le j\le n} \frac{W_{ij}}{\mu_j} \sum_{j=1}^n (\mu_j\eta_j)
% \end{aligned}
% \right. \nonumber \\
% \Rightarrow 
% &\left\{\begin{aligned}
% (Wy)_i &\ge \sum_{j=1}^n W_{ij}\beta_j +  r \min_{1\le j\le n} \frac{W_{ij}}{\mu_j} \ge b_l \\
% (Wy)_i &\le \sum_{j=1}^n W_{ij}\beta_j +  r \max_{1\le j\le n} \frac{W_{ij}}{\mu_j} \le b_u
% \end{aligned}
% \right. 
% \end{align}
Then, we have $b_l \mathbf{1}\le Wy\le b_u \mathbf{1}$ and 
\begin{align}
  F(y)=\frac{Wy}{\mu^\intercal W y}\ge \frac{ b_l \mathbf{1} }{\mu^\intercal (b_u \mathbf{1})}= \frac{b_l }{b_u}\mathbf{1} \ge \beta,
\end{align}
where the property $\mu^\intercal \mathbf{1}=1 $ is applied in the second equality. 
By induction, it follows that $y(k)=F(y(k-1))\in \Delta(\beta)$ for all $k\ge 0$. 
The proof is completed. $\hfill\square$
\end{pf}

Theorem \ref{th:F_invariance} gives a sufficient construction for $W$ to ensure that the state is always contained in the cone set $\Delta(\beta)$. 
Intuitively, \eqref{eq:lower_upper_bounds} has no direct dependence on the real-time state, and indicates that the change from $y_i(k)$ to $y_i(k+1)$ is bounded by $(b_u-b_l)$. 
This point corresponds to our common sense that the immune cells amounts will vary in a gradual way assuming a reasonable time-period \citep{brodin2017human}, e.g., 2h-20h or even a couple of weeks in between. 
We then present the following result.

\begin{theorem}[Convergence of $F$]\label{th:F_convergence} 
Under the conditions of Theorem \ref{th:F_invariance}, 
$y(k+1)=F(y(k))$ will converge to a unique fixed point $y^\star\in\Delta(\beta)$ satisfying 
\begin{align}\label{eq:fix}
Wy^\star=s y^\star~\text{with}~s=\mu^\intercal W y^\star. 
\end{align}
\end{theorem}

\begin{pf}
% We prove the two statements in this theorem by i) demonstrating the existence of the fixed point, and ii) deriving the converging point, respectively. 
% \noindent \textit{\textbf{Part 1: existence of the fixed point}}\\
% \indent 
To analyze the convergence of the model, we need to borrow some notions from nonlinear Perron-Frobenius theory \citep{Lemmens2012}. 
First, let $\mathcal{K}=\mathbb{R}^n_{>0}$ denote the interior of the closed positive cone $\mathbb{R}^n_{\ge0}$, and define the \emph{Hilbert projective metric} on $\mathcal{K}$ as\footnote{This metric is originally defined based on partially ordered vector spaces. Since this paper only focuses on the positive orthant $\mathbb{R}^n_{>0}$, we directly give its reduced form here.}
\begin{align}
d_H(x,y)
%\log\Big(\max_i\frac{x_i}{y_i}\cdot\max_i\frac{y_i}{x_i}\Big)
=\log\Big(\max_{i,j}\frac{x_i y_j}{y_i x_j} \Big),~x,y\in\mathcal{K} .
\end{align}
For a linear operator $L$ satisfying $L(\mathcal{K})\subset \mathcal{K}$, its projective diameter is defined by
\begin{align}
\delta(L)=\sup_{x,y\in\mathcal{K}} d_H(Lx,Ly).
\end{align}
Next, we introduce the following Birkhoff’s contraction lemma \citep[Theorem~2.9]{lemmens2014birkhoff}. 
\begin{lemma}\label{le:convergence}
If $L$ is a cone-linear mapping with $L(\mathcal{K})\subset\mathcal{K}$, then $L$ is a contraction in the Hilbert metric, 
satisfying 
\begin{align}\label{eq:contraction}
d_H(Lx,Ly)\le \kappa(L) \cdot d_H(x,y),
\end{align}
where $\kappa(L)=\tanh\Big(\frac{\delta(L)}{4}\Big)$ is the contraction ratio. 
% Specifically, if $\delta(W)<\infty$, then there exists a unique $v \in \mathcal{K}$ such that for any $y \in \mathcal{K}$, it holds that 
% \begin{align}
% d_H(W^k y,v)\le \kappa^k d_H(y,v).
% \end{align}
% % Specifically, if $\delta(W)<\infty$, then there exists a unique $v \in \mathcal{K}$ (normalized by $\|v\|=1$) such that 
% % \begin{align}
% % W v=r_W v,
% % \end{align}
% % where $r_W=\lim_{k \rightarrow \infty}\left\|W^k\right\|_{\mathcal{K}}^{1 / k}$ is the the cone spectral radius of $W$ and $\|W^k\|_{\mathcal{K}}=\sup \{\|W^k y\|: y \in \mathcal{K} \text { with }\|y\|=1\}$. 
% % For any $y \in \mathcal{K}$, it holds that 
% % \begin{align}
% % d_H(W^k y,v)\le \kappa^k d_H(y,v).
% % \end{align}
\end{lemma}

Note that the above Lemma was originally targeted at a linear mapping $L$, and we need to demonstrate how the constructed model $y(k+1)=F(y(k))=Wy(k)/(\mu^\intercal Wy(k))$ can sufficiently meet the conclusion in Lemma \ref{le:convergence}. 
First, notice that i) $\Delta(\beta)\subset \mathcal{K}$ is a compact subset in $\mathcal{K}$, and ii) $Wy\subset \Delta(\beta)$ always holds by assumption. 
Hence, the conclusion \eqref{eq:contraction} directly applies to the linear mapping $W$, i.e., 
\begin{align}\label{eq:contraction_W}
d_H(Wx,Wy)\le \kappa(W) \cdot d_H(x,y),~\forall x,y\in \Delta(\beta). 
\end{align}
Second, as the nonlinear mapping $F(y)$ only applies a normalization on $Wy$, 
it follows from the definition of $d_H$ that $\forall x,y\in\Delta(\beta)$,
\begin{align}
  d_H(F(x), F(y))&=d_H\left( \frac{Wx}{\mu^\intercal Wx}, \frac{Wy}{\mu^\intercal Wy}\right) \nonumber \\
  &=d_H( Wx, Wy)\le \kappa(W) d_H(x,y), 
  % d_H(\frac{y}{c}, v)=d_H(y, v),~y\in\mathcal{K},
\end{align}
which means that the normalization does not change the projective direction of the mapping $W$. 
Clearly, the mapping $F(y)$ inherits the same contraction ratio as \eqref{eq:contraction_W}. 

Finally, since $Wy$ is constrained by $b_l \mathbf{1}\le Wy\le b_u \mathbf{1}$, the projective diameter $\delta(W)$ is explicitly given by 
\begin{align}
\delta(W) &=\sup_{x,y\in\mathcal{K}} d_H(Wx,Wy) \nonumber \\
& = \sup_{x,y\in\mathcal{K}} \log\Big(\max_{i,j}\frac{(Wx)_i (Wy)_j}{ (Wy)_i (Wx)_j} \Big) \nonumber \\
&=\log\Big(\frac{b_u}{ b_l} \Big)^2=2\log\Big(\frac{b_u}{ b_l} \Big)<\infty. 
\end{align}
Thus, we have 
\begin{align}\label{eq:speed}
  \kappa(W)=\tanh\Big(\frac{\delta(W)}{4}\Big)=\tanh\Big(\frac{1}{2}\log\frac{b_u}{b_l}\Big)<1. 
\end{align}
By referring to the well-known Banach fixed-point theorem \citep{latif2013banach}, $(\Delta(\beta),d_H)$ is a non-empty complete metric space with a contraction mapping $F(\cdot)$, and thus $F(y(k))$ will converge to a unique fixed point $y^\star\in\Delta(\beta)$ such that 
\begin{align}
 F(y^\star)=\frac{Wy^\star}{\mu^\intercal W y^\star}=y^\star,
\end{align}
which leads to \eqref{eq:fix} and completes the proof. $\hfill\square$
\end{pf}

Theorem \ref{th:F_convergence} reveals that if the matrix $W$ satisfies the condition \eqref{eq:lower_upper_bounds}, the state $x(k)$ will converge to a fixed point that is determined by $W$. 
This property of $F(\cdot)$ is slightly different from the consensus model with a row-stochastic topology matrix, because the converging state of the latter is also dependent on the initial state. 
More importantly, considering that the ratio condition \eqref{eq:ratio_property} is expected to be met, if we suppose $y^\star=\alpha \mu$, then it follows that 
\begin{align}\label{eq:eigenvector}
W (\alpha \mu) = \left(\mu^\intercal W (\alpha \mu) \right ) (\alpha \mu) ~\Rightarrow ~W \mu =\alpha_{\mu} \mu, 
\end{align}
where $\alpha_{\mu}=\alpha \mu^\intercal W  \mu$. 
Clearly, here $\mu $ is a right eigenvector of $W$. 
Substituting \eqref{eq:eigenvector} into $y^\star=F(y^\star)$, we have
\begin{align}\label{eq:converge_state}
  y^\star&=\frac{W \mu}{\mu^\intercal W\mu} = \frac{\mu}{\mu^\intercal \mu} , \\
  x^\star&=\frac{\mu}{\mu^\intercal \mu}(\mu^\intercal x^\star)=\frac{(\mu^\intercal x(0))}{\mu^\intercal \mu} \mu,
\end{align}
where the property $\mu^\intercal x(k)=\mu^\intercal x(0) $ is applied in the second equality. 
Clearly, the stable cell amount distribution is directly dependent on $x(0)$ while exhibiting a ratio pattern.

\begin{remark}
Compared with the models M1)-M3), the proposed model has the following advantages regarding the feasibility conditions. 
i) There are no connectivity requirements and magnitude constraints on the eigenvalues of $W$ for stability concerns, only requiring $\mu$ to be a right eigenvector of $W$ by \eqref{eq:eigenvector}. 
ii) It is the cell percentage in the overall weighted amount sum, instead of the amount itself, that requires to bounded, which is more practical for the immune network setting.  
\end{remark}

\subsection{Method Design Under Limited Data}

With the modeling for the immune cell established, this part shows how to inversely infer the topology from limited data with measurement noises in an optimization framework.

Note that the stable reference profile $\mu$ and the bound parameters $\{\beta, b_l,b_u\}$ are given based on our prior knowledge and pre-experiments on the immune system. 
Considering the measurement noises, denote the data pair at $r$-th experiment as $\{\tilde{x}_0^{r},\tilde{x}_1^{r}\}$, normalized into 
% Denote the data pair of the measured cell amounts at $r$-th experiment as $\{\tilde{x}_0^{r},\tilde{x}_1^{r}\}$, and they are normalized into 
% Then, each measured data pair $\{\tilde{x}_0^{r},\tilde{x}_1^{r}\}$ is normalized into 
% and project it onto $\Delta(\beta)$:
\begin{align}
\tilde{y}_0^{r}=\frac{\tilde{x}_0^{r}}{\mu^\intercal \tilde{x}_0^{r}}, ~
\tilde{y}_1^{r}=\frac{\tilde{x}_1^{r}}{\mu^\intercal \tilde{x}_1^{r}}. 
\end{align}
We remark that the selection of $\beta$ is very conservative in practice, and thus $\tilde{y}_0^{r},\tilde{y}_1^{r}\ge \beta$ generally hold in the experiments. 
If not, we only need to further project $\tilde{y}_0^{r}$ and $\tilde{y}_1^{r}$ into $\Delta(\beta)$. 
Note that the data $\tilde{y}_1^{r}$ is desired to approximate $F(\tilde{y}_0^{r})$, or equivalently, $ (\mu^\intercal W \tilde{y}_0^{r}) \tilde{y}_1^{r}- W \tilde{y}_0^{r} \approx 0$, 
% \begin{align}
%  (\mu^\intercal W \tilde{y}_0^{r}) \tilde{y}_1^{r}- W \tilde{y}_0^{r} \approx 0,
% \end{align}
which is the residual error of the objective function in our optimization problem. 
In a vectorized form, this error contained in the data is given by
% it can be represented as 
\begin{align}
  \tilde{M}_D \cdot \operatorname{vec}(W) \triangleq [\tilde{M}_1^\intercal,\cdots,\tilde{M}_m^\intercal]^\intercal \cdot \operatorname{vec}(W),
\end{align}
where $\tilde{M}_r =\left[  \tilde{y}_1^{r} ( (\tilde{y}_0^{r})^\intercal \otimes\mu^\intercal ) \!-\! (\tilde{y}_0^{r})^\intercal \otimes I_n  \right],~r=1,\cdots,m$. 

Next, we demonstrate how to make $W$ meet the state positivity on $\Delta(\beta)$. 
Notice that in the condition \eqref{eq:lower_upper_bounds}, the function $c_l(i)$ is concave while $c_u(i)$ is convex regarding $W_{[i,:]}$. 
Hence, $c_l(i)\ge b_l$ and $c_b(i)\le b_u$ are all convex constraints. 
To facilitate solving the problem in a standard quadratic program (QP), we introduce two auxiliary variables $\{p_i,q_i\}$ and equivalently write \eqref{eq:lower_upper_bounds} as 
\begin{align}\label{eq:new_cons}
\left\{\begin{aligned}
    &\sum_{j=1}^n W_{ij}\beta_j+r p_i\ge b_l, ~ p_i\le \frac{W_{ij}}{\mu_j }, \forall j \\
  &\sum_{j=1}^n W_{ij}\beta_j+r q_i\le b_u, ~ q_i\ge \frac{W_{ij}}{\mu_j }, \forall j
\end{aligned}
\right.,
\end{align}
which adds $2n$ more constraints for each row but will not affect the feasibility. 
Based on Theorem \ref{th:F_invariance}, 
% when \eqref{eq:new_cons} and $({b_l }/{b_u})\mathbf{1} \ge \beta$ hold, 
when \eqref{eq:new_cons} and $\tilde{y}_0^{r}\in\Delta(\beta)$ hold, 
the properties $b_l \mathbf{1}\le W \tilde{y}_0^{r}\le b_u \mathbf{1}$ and $ F(\tilde{y}_0^{r})\in\Delta(\beta)$ will be satisfied automatically. 
In addition, recall that $y^\star=\frac{\mu}{\mu^\intercal \mu}$ is the fixed point of $F(\cdot)$ regardless of the data, 
and substituting it into $y^\star = F(y^\star)$ yields 
\begin{align}\label{eq:W_const}
   W\mu =s_{\mu} \mu ,
\end{align}
where $s_{\mu}=\frac{\mu^\intercal W \mu}{\mu^\intercal \mu}$. 
Hence, \eqref{eq:W_const} should be also treated as a strict constraint. 
% Equivalently, this constraint can be written as 
% \begin{align}\label{eq:mu_equality}
%   % b_l \mathbf{1} \le  s_{\mu}y^\star= W y^\star \le b_u \mathbf{1} ~ \Rightarrow ~  b_l\le s_{\mu} \le b_u,
%   W \mu= s_{\mu} \mu,~b_l\le s_{\mu} \le b_u,
% \end{align}
% % we require that $\mu$ is the fixed point of the mapping $F$. 
% % Left-multiplying $W \mu$ with $\mu^\intercal$, we have 
% % \begin{align}\label{eq:mu_equality}
% % b_l\le s_{\mu} =\mu^\intercal W \mu \le b_u,
% % \end{align}
% where $s_{\mu} = \frac{\mu^\intercal W \mu}{\mu^\intercal \mu} $, and the properties $\mu^\intercal y^\star=1 $ and $\mu^\intercal \bm{1}=1$ is used. 
% % Hence, the following constraint is also necessary
% % \begin{align}
% %   W\mu=s_{\mu} \mu~\text{for certain}~s_{\mu}\in[b_l,b_u]. 
% % \end{align}

Finally, based on the above formulation, inferring the topology $W$ from limited noisy data pairs is transformed to solving the following convex QP problem
\begin{subequations}\label{eq:op_problem}
\begin{align}
\!\! \min_{W,\{p_i,q_i\},s_{\mu}}~~ &  \|  \tilde{M}_D \cdot \operatorname{vec}(W) \|_2^2  + \gamma\| \operatorname{vec}(W)\|_1 \\
\text{s.t.}~~~~
% &\varepsilon\le s_r\le U, ~ r=1,\dots,m,\\
& W\mu =s_{\mu} \mu,\\
&\sum_{j=1}^n W_{ij}\beta_j+r p_i\ge b_l, ~ p_i\le \frac{W_{ij}}{\mu_j },~ \forall i,j,\\
&\sum_{j=1}^n W_{ij}\beta_j+r q_i\le b_u, ~ q_i\ge \frac{W_{ij}}{\mu_j }, ~ \forall i,j,
\end{align}
\end{subequations}
where $\gamma>0$ is the regularization parameter associated with $\| \operatorname{vec}(W)\|_1$. 
% Notice that the variable $s_{\mu}=\frac{\mu^\intercal W \mu}{\mu^\intercal \mu}$ does not affect the mapping result of $F(y)$ and the converging state $y^\star$, as revealed by \eqref{eq:converge_state}. 
% The equality $W\mu=s_{\mu}\mu$ essentially confines $W$ to an affine subspace. 
% However, different from the noise-free situation discussed in Remark \ref{rema:number},
% the equality $W\mu=s_{\mu}\mu$ confines $W$ to an affine subspace, and $s_{\mu}$ will affect the data residual and the regularization. 
% Hence, it is necessary to treat $s_{\mu}$ as a free variable to best reconcile the objective function, instead of fixing it as point in $[b_l, b_u]$. 

\begin{remark}
Note that introducing the additional $L_1$ norm term has two benefits. 
On the one hand, it could promote a sparse pattern on $W$, which corresponds to our common knowledge that one type of immune cell is directly influenced by only a few other cells. 
On the other hand, if only limited amount of experiment data is available (e.g., when $m<n$), it is very likely that the data matrix $\tilde{M}_D$ has small rank and renders no unique solution when only $\|  \tilde{M} \operatorname{vec}(W) \|_2^2$ is considered. 
Considering the introduced $L_1$ norm term and the randomness of $\tilde{M}_D$ in the objective function, the problem will have a unique solution with high probability (see \cite[Section 2]{tibshirani2013lasso} for details). 
\end{remark}

\section{Simulations and Experiments}\label{sec:simulation}

\subsection{Numerical Simulations}
First, we use an example network with $5$ nodes, whose topology matrix is given by 
\begin{align}\label{ex:network}
W =\begin{bmatrix}
0.4368 & 0.1690 & 0.9413 & 0.1539 & 0.2382 \\
0.1288 & 0.5313 & 0.4137 & 0.6771 & -0.0410 \\
0      & 0.4066 & -0.0418 & 0      & 1.4969 \\
0.1105 & 0.5494 & 0.2769 & 0.7737 & -0.0407 \\
1.1174 & 0      & 0.5472 & 0.0940 & 0.4353
\end{bmatrix}.
\end{align}

The ratio profile vector is $\mu= [0.22, 0.15, 0.23, 0.14, 0.26]^\intercal$, and the bound parameters are given by $\beta = \bm{1}/10$, $b_l  = 0.01$, and 
$b_u = 10$. 
It can be verified that this $W$ satisfies the conditions in Theorem \ref{th:F_invariance} and has a right eigenvector $\mu$. 
Since $\hat{W}$ is identifiable up a scalar ambiguity, we use a revised mean square error (denoted by $E(W,\hat{W})$) and the ratio of correctly inferred edge signs (denoted by $R(W,\hat{W})$) to evaluate the inference performance on the topology, 
% The error is given by 
\begin{align}
  &E(W,\hat{W})= \frac{\min_{\alpha>0} \| W- \alpha \hat{W}\|_{Frob}^2} {\| W \|_{Frob}^2 }, \\
  &R(W,\hat{W})=1-\frac{ \|\operatorname{sign}(W) - \operatorname{sign}(\hat{W}) \|_0  }{n^2}. 
\end{align}
The smaller $E(W,\hat{W})$ and higher $R(W,\hat{W})$ indicate better inference performance. 

Given the initial state $x(0)=[345,75,1200,345,457]^\intercal$, the simulated results are given in Fig.~\ref{fig:numerical}. 
The evolution of ratios of different components is plotted in Fig.~\ref{fig:ratio}, where solid and dash lines correspond to the actual and desired ratios, respectively. 
It is clear that the ratios of the components will converge to $\mu$ as the iteration increases. 
Then, using different amount of data pairs, the topology matrix is obtained by solving the homogeneous equation $M \operatorname{vec}(W)=0$. 
As shown in Fig.~\ref{fig:error}, the corresponding $E(W,\hat{W})$ and $R(W,\hat{W})$ are biased when $m<5$, while perfect when $m\ge 5$. 
This phenomenon matches our intuition that the topology is identifiable up to a scalar when we have an appropriate number of noise-free data.

\begin{figure}[t]
\centering
\subfigure[State ratio evolution of different components.]{\label{fig:ratio}
% \vspace{-10pt}
\includegraphics[width=0.36\textwidth]{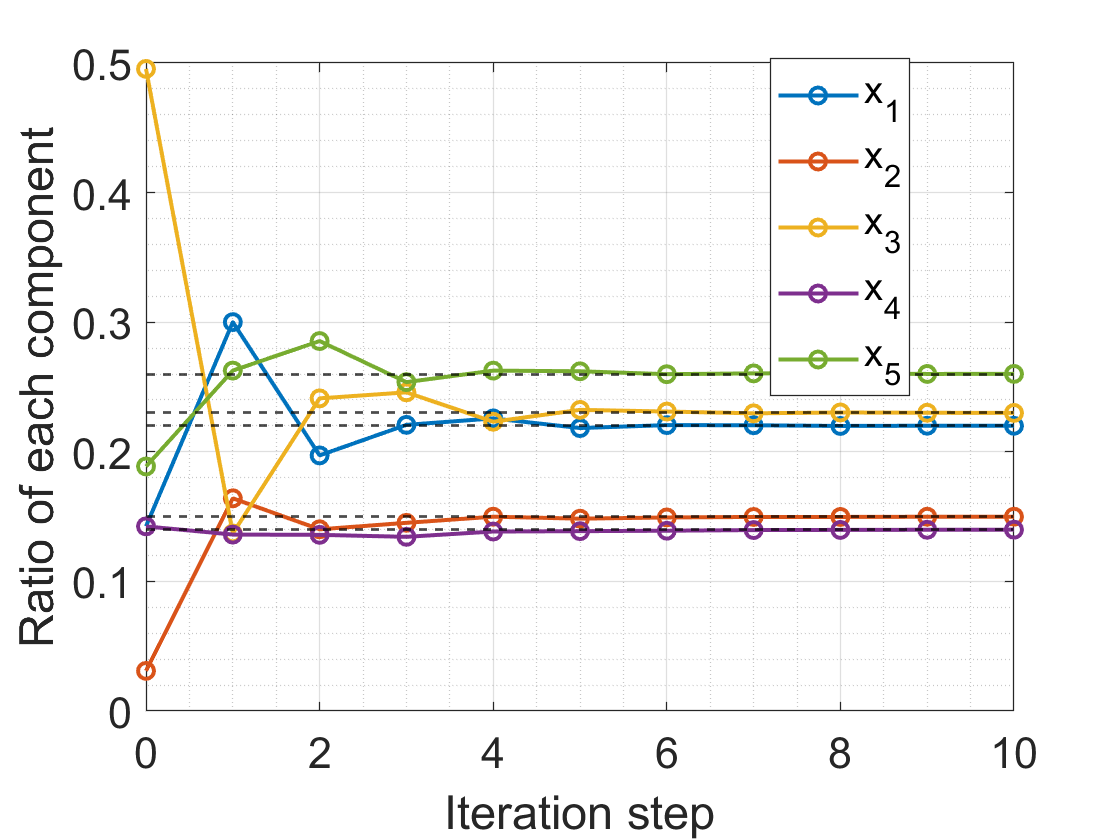}}
% \vspace{-10pt}
\subfigure[Inference error of $\hat{W}$.]{\label{fig:error}
\includegraphics[width=0.36\textwidth]{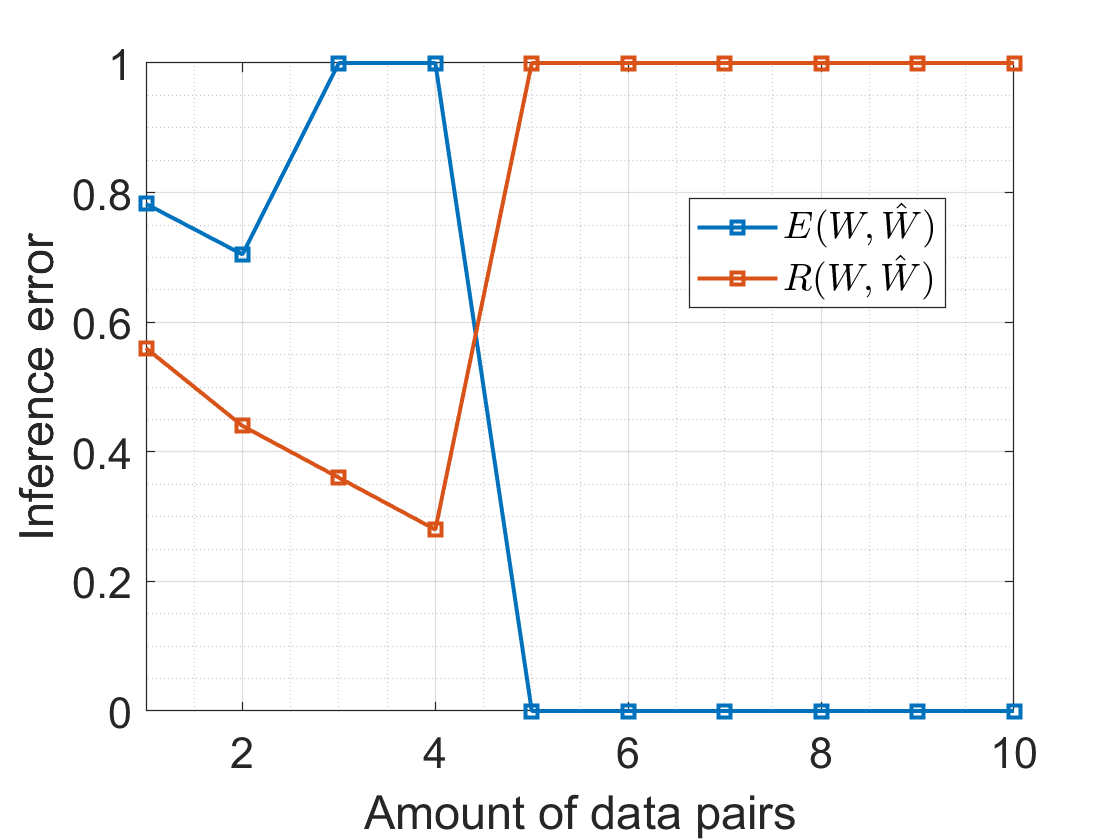}}
\vspace{-8pt}
\caption{Inference performance on a numerical example.}
\label{fig:numerical}
% \vspace{-8pt}
\end{figure}

\begin{figure*}[t]
\centering
\subfigure[Cell distribution at $2$h and $20$h in one experiment.]{\label{fig:cell}
% \vspace{-10pt}
\includegraphics[width=0.38\textwidth]{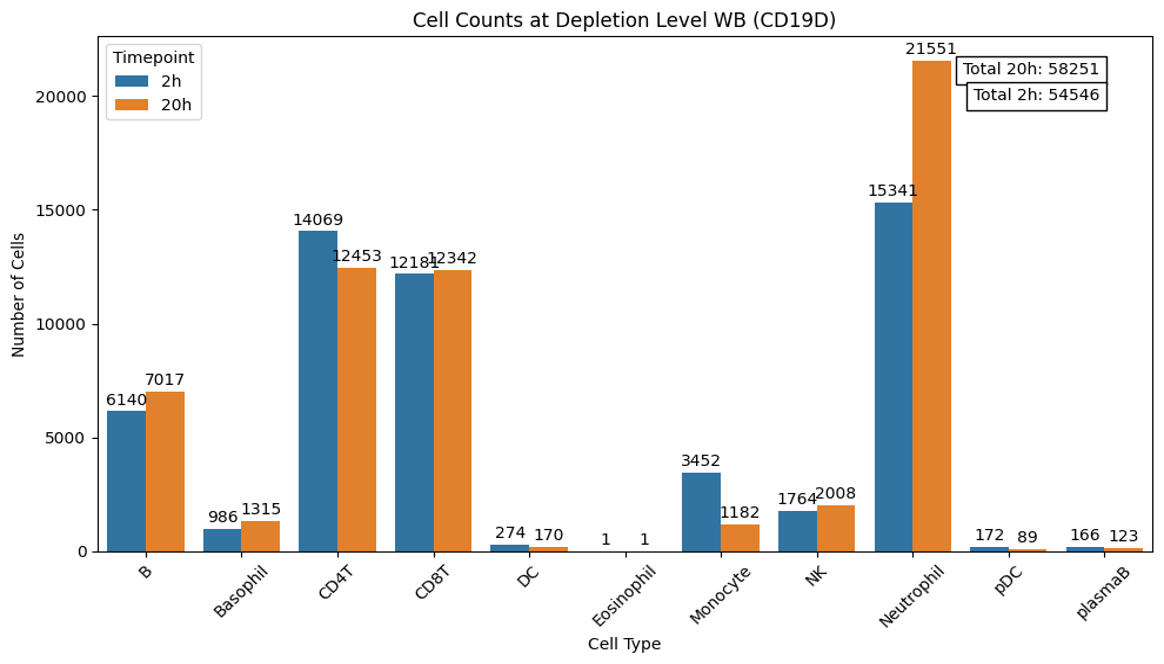}}
\subfigure[Inferred topology matrix illustration.]{\label{fig:topo}
% \hspace{-10pt}
\includegraphics[width=0.29\textwidth]{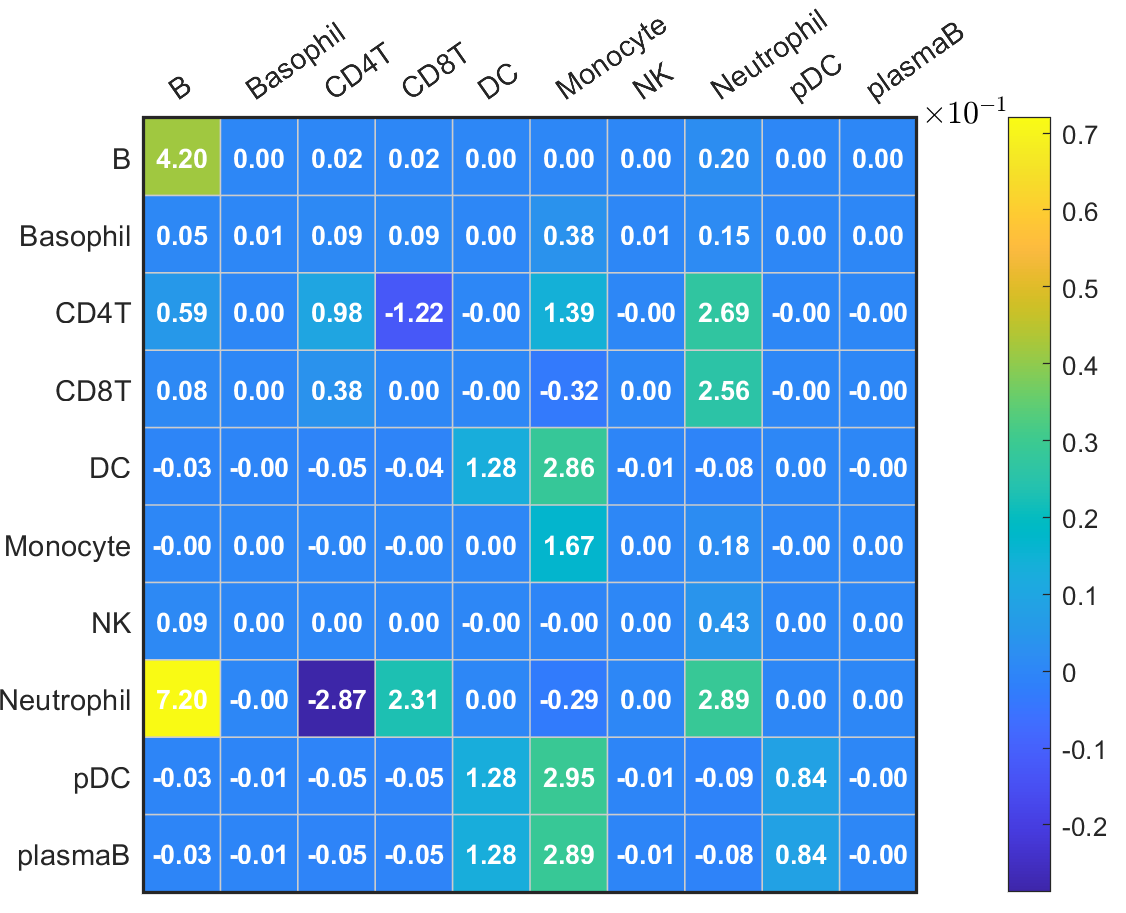}}
% \hspace{-10pt}
\subfigure[Prediction errors on data pairs.]{\label{fig:prediction}
\includegraphics[width=0.28\textwidth]{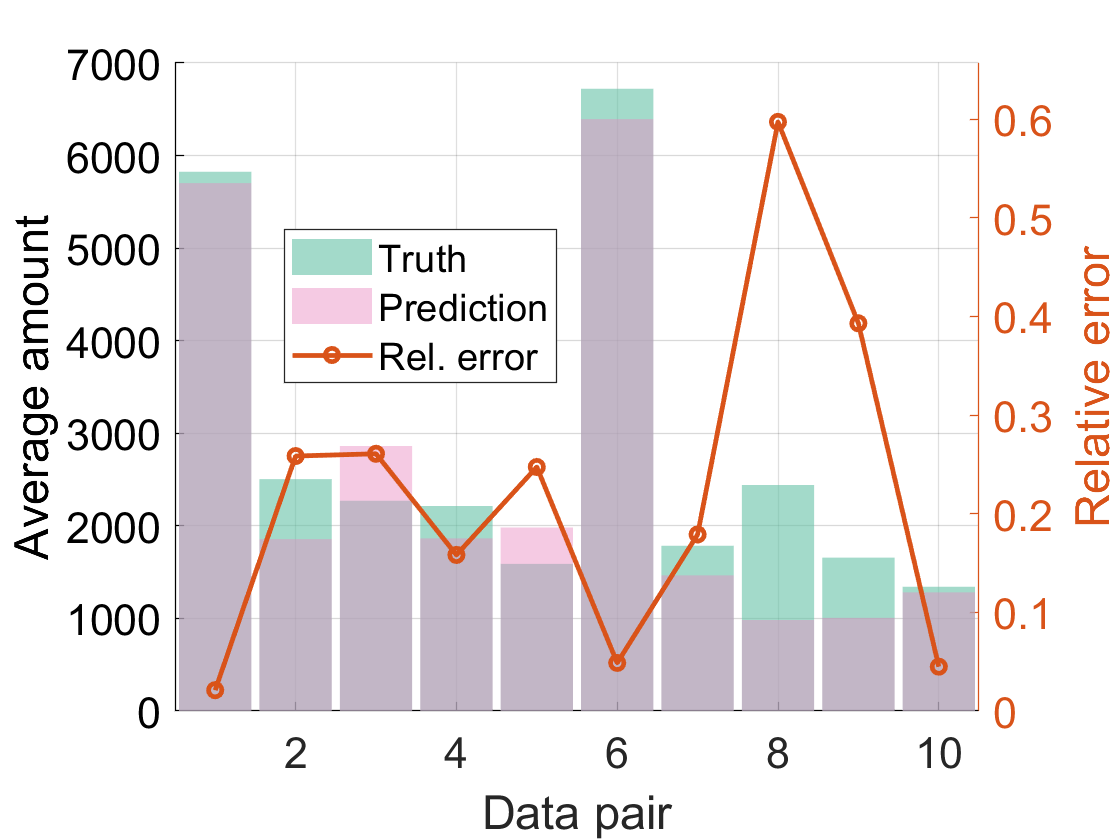}}
\vspace{-8pt}
\caption{Inference performance on real experiments.}
\label{fig:experiments}
% \vspace{-8pt}
\end{figure*}

% \begin{figure}[t]
% \centering
% \subfigure[Graph illustration of the inferred topology.]{\label{fig:topo}
% % \hspace{-10pt}
% \includegraphics[width=0.4\textwidth]{Experiment_matrix.png}}
% % \hspace{-10pt}
% \subfigure[Prediction errors on data pairs.]{\label{fig:prediction}
% \includegraphics[width=0.4\textwidth]{Experiment_prediction.png}}
% \vspace{-8pt}
% \caption{Inference performance on real experiments.}
% \label{fig:experiments}
% % \vspace{-8pt}
% \end{figure}

\subsection{Validation on Real Immune Cell Experiments}
 
In the real immune cell experiments (Forlin et al, Unpublished), we consider $10$ types of immune cells: \{\textit{B, Basophil, CD4T, CD8T, DC, Monocyte, NK, Neutrophil, pDC, plasmaB}\}, and order them from $1$ to $n$. 
The procedures of these experiments are explained in Sec. \ref{subsec:principle}. 
% and have collected $m=10$ groups of data pairs. 
Here we collect $m=10$ groups of data pairs, and draw the cell amount distribution of one group in Fig.~\ref{fig:cell}. 
Considering the measurement noises, we obtain the topology by solving the optimization problem \eqref{eq:op_problem}, and visualize it in Fig.~\ref{fig:topo}. 
Note that the values displayed in the matrix blocks are magnified $10$ times for better reading. 
Since the ground truth topology of the immune network is unknown in practice, we cannot use the metrics $E(W,\hat{W})$ and $R(W,\hat{W})$ to directly evaluate the inference performance. 
Instead, we adopt the following relative prediction error on a sample
\begin{align}\label{eq:rela_error}
  E_p (\hat{W},x_0)=\left\|\frac{  \mu^\intercal x_0 }{ \mu^\intercal \hat{W} x_0 }  \hat{W} x_0  - x_1 \right \|/\|x_1 \|.
\end{align}
Then, the average amount of all components of $x^r_1(r=1,\cdots,10)$ and its corresponding prediction amount are drawn in Fig.~\ref{fig:prediction}, along with the relative prediction error curve. 
It is intuitive to find that most of the relative errors are below $0.3$. 
Notice that the average of the 10 relative errors is $0.327$, with only two of them being larger than $0.3$. 
These two results correspond to the cases where the overall cell amount is small. 
In this regard, the proposed model and inference method achieve acceptable performance on revealing the interaction topology of the immune network.

\section{Conclusions}\label{sec:conclusion}

In this paper, we investigated the topology inference problem of a class of immune networks. 
First, 
% based on certain common knowledge on the interactions of immune cells, 
we constructed a new nonlinear model to describe the immune network, enjoying the merits of simple structure and physical interpretations. 
Then, we derived the sufficient conditions for model to guarantee that the state non-negativity and the ratio-based convergence can be achieved simultaneously. 
Finally, a constrained QP method was provided to infer the topology matrix from data. 
Numerical simulations and validation on experiment data demonstrated the effectiveness of the proposed method. 
Future direction includes investigating the topology identifiablity under weak prior parameter assumptions, giving systematic sensitivity analysis to prior parameters, and providing efficient input design under stimulants for practical immune experiments.

% \bibliography{ifacconf.bib}             % bib file to produce the bibliography

                                                     % with bibtex (preferred)

\end{document}